\documentclass[english, a4paper]{article}

\usepackage[margin=1.3in]{geometry}
\usepackage[font=small, labelfont=bf, margin=0.35in]{caption}

\usepackage{babel}
\usepackage[utf8]{inputenc}
\usepackage[pdftex]{hyperref}

\usepackage{titling}
\usepackage{authblk}

\usepackage{cite}

\usepackage{amsmath}
\usepackage{amsthm}
\usepackage{amssymb}

\newtheorem{theorem}{Theorem}
\newtheorem{lemma}{Lemma}
\newtheorem{claim}{Claim}
\newtheorem{definition}{Definition}

\usepackage{enumitem}

\newlist{pseudocode}{enumerate}{3}
\setlist[pseudocode]{
  label={\arabic*},
  before=\raggedright,
  leftmargin=*,
  itemsep=3pt,
  topsep=3pt}
\setlist[pseudocode,2]{label*={.\arabic*}}
\setlist[pseudocode,3]{label*={.\arabic*}}

\newlist{romanlist}{enumerate}{1}
\setlist[romanlist]{
  label={\roman*)},
  leftmargin=*,
  itemsep=3pt,
  topsep=3pt}

\newcommand{\ordo}{O}
\newcommand{\lcm}{\mathrm{lcm}}

\title{On completely factoring any integer efficiently in a single run of an order finding algorithm}
\author[1,2]{\href{mailto:ekera@kth.se}{Martin Ekerå}}
\affil[1]{\small KTH Royal Institute of Technology, Stockholm, Sweden}
\affil[2]{\small Swedish NCSA, Swedish Armed Forces, Stockholm, Sweden}

\predate{\begin{center}}
\postdate{\end{center}\vspace{-3mm}}

\allowdisplaybreaks

\begin{document}
\maketitle

\begin{abstract}
  We show that given the order of a single element selected uniformly at random from $\mathbb Z_N^*$, we can with very high probability, and for any integer $N$, efficiently find the complete factorization of $N$ in polynomial time.
  This implies that a single run of the quantum part of Shor's factoring algorithm is usually sufficient.
  All prime factors of $N$ can then be recovered with negligible computational cost in a classical post-processing step.
  The classical algorithm required for this step is essentially due to Miller.
\end{abstract}

\section{Introduction}
In what follows, let
\begin{align*}
N = \prod_{i \, = \, 1}^{n} p_i^{e_i}
\end{align*}
be an $m$ bit integer, with $n \ge 2$ distinct prime factors $p_i$, for $e_i$ some positive exponents.

Let an algorithm be said to \emph{factor} $N$ if it computes a non-trivial factor of $N$, and to \emph{completely} factor $N$ if it computes the set $\{ p_1, \ldots, p_n \}$.
Let~$\phi$ be Euler's totient function, $\lambda$ be the Carmichael function, and $\lambda'(N) = \lcm(p_1 - 1, \ldots, p_n - 1)$.
Furthermore, let~$\mathbb Z_N^*$ denote the multiplicative group of~$\mathbb Z_N$, the ring of integers modulo~$N$, and let $\ln$ and $\log$ be the natural and base two logarithms, respectively.
Denote by~$[a, b]$ the integers from~$a$ up to an including~$b$.

Throughout this paper, we shall assume $N$ to be odd for proof-technical reasons.
This does not imply a loss of generality:
It is easy to fulfill this requirement by using trial division.
Indeed, one would in general always remove small prime factors before calling upon more elaborate factoring algorithms.
Note furthermore that order finding may be performed prior to trial division being applied to $N$ if desired, see section~\ref{section:multiple-N} for further details.

There exists efficient probabilistic primality tests, such as Miller-Rabin~\cite{miller, rabin}, and efficient algorithms for reducing perfect powers $z = q^e$ to $q$: A simple option is to test if $z^{1/d}$ is an integer for some $d \in [2, \lfloor \log z \rfloor]$.
For more advanced options, see e.g.~Bernstein et al.~\cite{bernstein}.

\section{Earlier works}
Shor~\cite{shor1994, shor1997} proposed to factor $N$ by repeatedly selecting a random $g \in \mathbb Z_N^*$, computing its order $r$ via quantum order finding, and executing a classical procedure inspired by Miller~\cite{miller}.
Specifically, Shor proposed to use that if $r$ is even, and $g^{r/2} \not\equiv -1 \:\: (\text{mod } N)$, it must be that
\begin{align*}
(g^{r/2} - 1)(g^{r/2} + 1) = g^r - 1 \equiv 0 \quad (\text{mod } N)
\end{align*}
so $\gcd((g^{r/2} \pm 1) \text{ mod } N, N)$ yields non-trivial factors of $N$.
Note that $g^{r/2} \not\equiv 1 \:\: (\text{mod } N)$ by definition, as $r$ is otherwise not the order of $g$.
Shor proved that the probability of the above two requirements being met is at least $1/2$.
If both requirements are not met, the algorithm may be re-run for a new $g$, in which case the probability is again at least $1/2$ of succeeding.

This implies that Shor's algorithm will eventually succeed in finding two non-trivial factors of $N$.
However, as re-running the quantum order finding part of the algorithm is expensive, it is natural to consider improved strategies.
To completely factor $N$, recursive calls to Shor's factoring algorithm, and hence to the quantum order finding part, would naïvely be required, albeit with consecutively smaller factors, until the factors are prime, perfect powers, or sufficiently small to factor using classical algorithms.
Again, it is natural to consider improved strategies to avoid re-runs in this setting.

\subsection{On the success probability of quantum order finding}
Shor's factoring algorithm as originally described can fail either because the order~$r$ of~$g$ is not amenable to factoring~$N$, in the sense that $r$ is odd or $g^{r/2} \equiv -1 \:\: (\text{mod } N)$, or because the order finding part of the algorithm fails to return~$r$ given~$g$.

The probability of the algorithm failing for the latter reason is negligible, however, if the quantum part is correctly parameterized and post-processed, and if it is executed as per its mathematical description by the quantum computer, see e.g.~Appendix A to~\cite{ekera-general} or~\cite{bourdon}.

In what follows, we therefore primarily focus our attention on classically recovering non-trivial factors of $N$ given $r$.
When referring to factoring in a single run of an order finding algorithm, we assume the order finding algorithm to yield $r$ given $g$.

\subsection{Tradeoffs in quantum order finding}
Seifert~\cite{seifert} has proposed to modify the order finding part of Shor's algorithm to enable tradeoffs between the number of runs that need be performed on the quantum computer, and the complexity of each run.
In essence, Seifert's idea is to compute only partial information on the order in each run, thereby reducing the number of operations that need to be performed by the computer in each run without loss of coherence.
Given the outputs from a sufficiently large number of such partial runs, the order may then be re-constructed efficiently classically, yielding a complete order finding algorithm that returns $r$ with high probability given $g$.

Making tradeoffs may prove advantageous in the early days of quantum computing when the capabilities of the computers available are limited.
On a side note, Knill~\cite{knill} has proposed a different kind of tradeoffs, where the goal is not to perform fewer operations in each run, but rather to obtain improved lower bounds on the success probability of the order being returned.

\subsection{Improvements for odd orders}
Several improvements to Shor's original classical post-processing approach have been proposed, including in particular ways of recovering factors of $N$ from odd orders~\cite{martin-lopez, lawson, grosshans, johnston}.
Grosshans et al.~\cite{grosshans} point out that if a small prime factor $q$ divides $r$, then $\gcd((g^{r/q} - 1) \text{ mod } N, N)$ is likely to yield non-trivial factors of $N$.
Johnston~\cite{johnston} later made similar observations.

In the context of Shor's algorithm, the observation that odd~$r$ may yield non-trivial factors of~$N$ seems to first have been made by Martín-López et al.~\cite{martin-lopez} in an actual experimental implementation.
This is reported in a work by Lawson~\cite{lawson} and later by Grosshans et al.~\cite{grosshans}.

We may efficiently find all small and moderate prime factors of $r$.
This often gives us several attempts at recovering non-trivial factors of $N$, leading to an increase in the probability of factoring $N$.
Furthermore, we can try all combinations of these prime factors, with multiplicity when applicable, to increase the number of non-trivial divisors.

\subsection{Improvements for special form integers}
\label{section:special-form}
Ekerå and Håstad~\cite{ekera-hastad, ekera-pp} have introduced a specialized quantum algorithm for factoring RSA integers that is more efficient than Shor's general factoring algorithm.
The problem of factoring RSA integers merits special consideration because it underpins the security of the widely deployed RSA crypto\-system~\cite{rsa}.

The algorithm of Ekerå and Håstad classically reduces the RSA integer factoring problem to a short discrete logarithm problem in a cyclic group of unknown order, using ideas from~\cite{hss93}, and solves this problem quantumly.
It is more efficient primarily because the quantum step is less costly compared to traditional quantum order finding, both when not making tradeoffs and comparing to Shor, and when making tradeoffs and comparing to Seifert.
It furthermore allows for the two factors of the RSA integer to be recovered deterministically, once the short discrete logarithm has been computed.
This implies that there is little point in optimizing the post-processing in Shor's original algorithm if the goal is to factor RSA integers.

On the topic of factoring special form integers, Grosshans et al.~\cite{grosshans} have shown how so-called safe semi-primes may be factored deterministically after a single run of Shor's original order finding algorithm.
Xu et al.~\cite{xu} have presented similar ideas.
Leander~\cite{leander} has shown how the lower bound of $1/2$ in Shor's original analysis may be improved to $3/4$ for semi-primes.

\subsection{Other related works on factoring via order finding}
There is a considerable body of literature on factoring.
The specific problem of factoring via number theoretical oracles has been widely explored, in the scope of various contexts.
Many of the results have a lineage that can be traced back to the seminal works of Miller~\cite{miller}.

More recently, Morain et al.~\cite{morain} have investigated \emph{deterministic} algorithms for factoring via oracles that yield $\phi(N)$, $\lambda(N)$ or the order $r$ of an element $g \in \mathbb Z_N^*$.
They find that given $\phi(N)$, it is possible to factor $N$ unconditionally and deterministically in polynomial time, provided that certain conditions on the prime factors of $N$ are met: It is required that $N$ be square-free and that $N$ has a prime factor $p > \sqrt{N}$.
Their approach leverages the Lenstra-Lenstra-Lovász (LLL)~\cite{lll} lattice basis reduction algorithm.

Morain et al. furthermore explicitly note that their work is connected to Shor's factoring algorithm, and that efficient \emph{randomized} factoring algorithms are produced by all three oracles (see sections 2.3 and 2.5 in~\cite{morain}).
They recall the method of Miller~\cite{miller} for factoring via an oracle that yields the order $r$ of $g \in \mathbb Z_N^*$, and its use in Shor's algorithm, and the fact that it may be necessary to consider multiple $g$ to find one with even order suitable for factoring $N$.
This implies that multiple oracle calls may be required to find non-trivial factors.

The authors furthermore state that if one has access to an oracle that yields e.g.~$\phi(N)$ or $\lambda(N)$, it is possible to do better:
It is then possible to find a $g \not\equiv \pm 1 \:\: (\text{mod } N)$ such that $g^2 \equiv 1 \:\: (\text{mod } N)$.
In particular, one may use that the order of $g$ must divide $\phi(N) = 2^t o$ for some $t > 0$ and odd $o$ to efficiently find such $g$.
This is closely related to the observations made in this paper.
The original algorithm is from Miller~\cite{miller}.

\subsection{On the relation to our contribution}
Given the abundance of literature on the topic of factoring, it is admittedly hard to make new original contributions, or even to survey the existing literature in its entirety.

We are however not aware of anyone previously demonstrating, within the context of Shor's algorithm, that a single call to the order finding algorithm is in general sufficient to completely factor any composite integer with high probability.
On the contrary, it is sometimes said that Shor's original post-processing algorithm should be used, potentially requiring several runs to find even a non-trivial factor, let alone the complete factorization.

\section{Our contribution}
We give an efficient classical probabilistic polynomial time algorithm, that is essentially due to Miller~\cite{miller}, for completely factoring $N$ given the order $r$ of a single element $g$ selected uniformly at random from $\mathbb Z_N^*$.
We furthermore analyze the runtime and success probability of the algorithm: In particular, we give a lower bound on its success probability.

\subsection{Notes on our original intuition for this work}
Given the order $r$ of $g$, we can in general correctly guess the orders of a large fraction of the other elements of $\mathbb Z_N^*$ with high probability.
To see why this is, note that $g$ is likely to have an order such that $\lambda(N) / r$ is a moderate size product of small prime factors.
Hence, by multiplying on or dividing off small prime factors to $r$, we can guess $\lambda(N)$, and by extension the orders of other elements in the group.

The above observation served as our original intuition for pursuing this line of work.
In this paper, we do however take matters a few steps further:
In particular, instead of guessing the orders of individual elements in $\mathbb Z_N^*$, we instead guess some positive multiple of $\lambda'(N)$.
Furthermore, we show that even if we only manage to guess some positive multiple of a divisor of $\lambda'(N)$, we are still often successful in recovering the complete factorization of $N$.

\subsection{The algorithm}
\label{section:algorithm}
In what follows, we describe a classical algorithm, essentially due to Miller~\cite{miller} (see the algo\-rithm in Lemma 5 ERH), for completely factoring $N$ given a positive multiple of $\lambda'(N)$.

We have slightly modified it, however, by adding a step in which we attempt to guess such a multiple, denoted~$r'$ below, given the order~$r$ of~$g$.
Furthermore, we select~$k$ group elements~$x_j$ uniformly at random from~$\mathbb Z_N^*$, for $k \ge 1$ some small parameter that may be freely selected, whereas Miller iterates over all elements up to some bound.

With these modifications, we shall prove that the resulting probabilistic algorithm runs in polynomial time, with the possible exception of the call to an order finding algorithm in the first step, and analyze its success probability.
To be specific, the resulting algorithm first executes the below procedure once to find non-trivial factors of $N$:

\begin{pseudocode}
\item Select $g$ uniformly at random from $\mathbb Z_N^*$. \label{alg:step:select-g}

Compute the order $r$ of $g$ via an order finding algorithm.

\item Let $\mathcal P(B)$ be the set of primes $\le B$. \label{alg:step:grow-r-to-rp}

Let $\eta (q, B)$ be the largest integer such that $q^{\eta(q, B)} \le B$.

Let $m' = cm$ for some constant $c \ge 1$ that may be freely selected.

Compute $r' = r \prod_{q \, \in \, \mathcal P(m')} q^{\eta(q, m')}$. \label{alg:step:guess}

\item Let $r' = 2^t o$ where $o$ is odd.

\item For $j = 1, \, 2, \, \ldots, \, k$ for some $k \ge 1$ that may be freely selected do:
\label{algorithm-step:for-loop}
\begin{pseudocode}
\item Select $x_j$ uniformly at random from $\mathbb Z_N^*$. \label{alg:step:select-xj}

\item For $i = 0, \, 1, \, \ldots, \, t$ do:
\begin{pseudocode}
\item Compute $d_{i, j} = \gcd(x_j^{2^i o} - 1, N)$. \label{alg:step:take-gcd}

If $1 < d_{i, j} < N$ report $d_{i, j}$ as a non-trivial factor of $N$.
\end{pseudocode}
\end{pseudocode}
\end{pseudocode}
We then obtain the complete factorization from the $d_{i, j}$ reported as follows:

A set is initialized and $N$ added to it before executing the above algorithm.
For each non-trivial factor reported, the factor is added to the set.
The set is kept reduced, so that it contains only non-trivial pairwise coprime factors.
It is furthermore checked for each factor in the set, if it is a perfect power $q^e$, in which case $q$ is reported as a non-trivial factor.
The algorithm succeeds if the set contains all distinct prime factors of $N$ when the algorithm stops.

Recall from the introduction that there are efficient methods for reducing $q^e$ to $q$, and methods for testing primality in probabilistic polynomial time.

\subsubsection{Notes on efficient implementation}
Note that the algorithm as described in section~\ref{section:algorithm} is not optimized:
Rather it is presented for ease of comprehension and analysis.
In an actual implementation, it would for example be beneficial to perform arithmetic modulo $N'$ throughout step~\ref{algorithm-step:for-loop} of the algorithm, for $N'$ a composite divisor of $N$ that is void of prime factors that have already been found.

The algorithm would of course also stop incrementing $j$ as soon as the factorization is complete, rather than after $k$ iterations, and stop incrementing $i$ as soon as $x_j^{2^i o} \equiv 1 \:\: (\text{mod } N')$ rather than continue up to $t$.
It would select $x_j$ and $g$ from $\mathbb Z_N^* \backslash \{ 1 \}$ in steps~\ref{alg:step:select-g} and~\ref{alg:step:select-xj} rather than from $\mathbb Z_N^*$.
In step~\ref{alg:step:take-gcd}, it would compute $d_{i,j} = \gcd(u_{i,j} - 1, N')$, where $u_{0,j} = x_j^o \text{ mod } N'$, and $u_{i,j} = u_{i-1,j}^2 \text{ mod } N'$ for $i \in [1, t]$, to avoid raising $x_j$ to $o$ repeatedly.

\vspace{1.5mm}

\emph{Ideas for potential optimizations:}
To further speed up the exponentiations, instead of raising each $x_j$ to a pre-computed guess $r'$ for a multiple of $\lambda'(N)$, a smaller exponent that is a multiple of the order of~$x_j \text{ mod } N'$ may conceivably be speculatively computed and used in place of $r'$.
To obtain more non-trivial factors from each $x_j$, combinations of small divisors of the exponent may conceivably be exhausted; not only the powers of two that divide the exponent.

\vspace{1.5mm}

\emph{Missing factors:} 
If an $x_j$ is such that $w_j = x_j^{N' r'} \not\equiv 1 \:\: (\text{mod } N')$, a factor~$q$ equal to the order of $w_j \text{ mod } N'$ is missing in the guess $r'$ for a multiple of $\lambda'(N')$.
Should this lead the algorithm to fail to completely factor $N'$, it may be worthwhile to attempt to compute the missing factor:
The options available include searching for $q$ by exponentiating to all primes up to some bound, which is essentially analogous to increasing $c$, or using some form of cycle finding algorithm that does not require a multiple of $q$ to be known in advance.  

\vspace{1.5mm}

In short, there are a number of optimizations that may be applied, but doing so above would obscure the workings of the algorithm, and the analysis that we are about to present. It is furthermore not necessary, since the algorithm as described is already very efficient.

\subsubsection{Notes on performing order finding for a multiple of $N$}
\label{section:multiple-N}
Note that the order finding in step~\ref{alg:step:select-g} may be performed for a multiple of~$N$ if desired.
This can only cause~$r$ to grow by some multiple, which in turn can only serve to increase the success probability, in the same way that growing~$r$ to~$r'$ in step~\ref{alg:step:grow-r-to-rp} serves to increase the success probability, see section~\ref{section:analysis}.
In turn, this explains why we can replace $N$ with $N'$ as described in the previous section, and why a restriction to odd $N$ does not imply a loss of generality.

\subsubsection{Notes on analogies with Miller's and Rabin's works}
Miller's original version of the algorithm in section~\ref{section:algorithm} is deterministic, and proven to work only assuming the validity of the extended Riemann hypothesis (ERH), as is Miller's primality test in the same thesis~\cite{miller}.
This may be because the notion of probabilistic polynomial time algorithms was not universally recognized when Miller did his thesis work.

Rabin~\cite{rabin} later converted Miller's primality test algorithm into a probabilistic polynomial time algorithm that is highly efficient in practice.

It is perhaps interesting to note that we perform essentially the same conversion in section~\ref{section:algorithm} with respect to Miller's factoring algorithm:
We convert it into an efficient probabilistic polynomial time factoring algorithm, that recovers the complete factorization of $N$ given the order $r$ of $g$ selected uniformly at random from $\mathbb Z_N^*$.

\subsubsection{Notes on analogies with Pollard's works}
Miller's algorithm may be regarded as a generalization of Pollard's $p-1$ algorithm~\cite{pollard-p-minus-one}: Miller essentially runs Pollard's algorithm for all prime factors $p_i$ in parallel by using that a multiple of $\lambda'(N) = \lcm(p_1 - 1, \ldots, p_n - 1)$ is known. Pollard, assuming no prior knowledge, uses a product of small prime powers up to some smoothness bound $B$ in place of $\lambda'(N)$. This factors out $p_i$ from $N$ if $p_i-1$ is $B$-smooth, giving Pollard's algorithm its name.

Since we only know $r$, a multiple of some divisor of $\lambda'(N)$, we grow $r$ to $r'$ by multiplying on a product of small prime powers.
This is in analogy with Pollard's approach.

\subsection{Analysis of the algorithm}
\label{section:analysis}
A key difference between our modified algorithm in section~\ref{section:algorithm}, and the original algorithm in Miller's thesis~\cite{miller}, is that we compute the order of a random $g$ and then add a guessing step:

We guess an $r'$ in step~\ref{alg:step:guess} that we hope will be a multiple of $p_i - 1$ for all $i \in [1, n]$, and if not all, then at least for all but one index on this interval, in which case the algorithm will still be successful in completely factoring $N$.
This is shown in the below analysis.
Specifically, we lower-bound the success probability and demonstrate the polynomial runtime of the algorithm.

\label{section:success-probability}
\begin{definition}
The prime $p_i$ is unlucky if $r'$ is not a multiple of $p_i - 1$.
\end{definition}

\begin{lemma}
\label{lemma:probability-pi-unlucky}
The probability that $p_i$ is unlucky is at most $\log p_i / (m' \log m')$.
\end{lemma}
\begin{proof}
For $p_i$ to be unlucky, there has to exist a prime power $q^e$ such that 
\begin{romanlist}
\item $q^e > m'$, as $q^e$ otherwise divides $r'$,
\item $q^e$ divides $p_i - 1$, and
\item $g$ is a $q^e$-power modulo $p_i$, to reduce the order of $g$ by a factor $q^e$.
\end{romanlist}
The number of such prime powers $q^e$ that divide $p_i - 1$ is at most $\log p_i / \log m'$, as $q^e > m'$ and the product of the prime powers in question cannot exceed $p_i - 1$.
For each such prime power, the probability that $g$ is a $q^e$-power is at most $1 / q^{e} \le 1/m'$.
The lemma follows by taking the product of these two expressions.
\end{proof}

\begin{lemma}
\label{lemma:at-most-one-prime-unlucky}
If at most one $p_i$ is unlucky, then except with probability at most
\begin{align*}
2^{-k} \cdot \binom{n}{2}
\end{align*}
all $n$ prime factors of $N$ will be recovered by the algorithm after $k$ iterations.

\end{lemma}
\begin{proof}
For us not to find all prime factors,
there must exist two distinct prime factors $q_{1}$ and $q_{2}$ that both divide $N$, such that for all combinations of $i \in [0, t]$ and $j \in [1, k]$, either both factors divide $x_j^{2^i o}- 1$, or none of them divide $x_j^{2^i o}- 1$.

To see why this is, note that the two factors will otherwise be split apart for some combination of $i$ and $j$ in step~\ref{alg:step:take-gcd} of the algorithm in section~\ref{section:algorithm}, and if this occurs pairwise for all factors, the algorithm will recover all factors.

There are $\binom{n}{2}$ ways to select two distinct primes from the $n$ distinct primes that divide~$N$.
For each such pair, at most one of $q_1$ and $q_2$ is unlucky, by the formulation of the lemma.
\begin{romanlist}[parsep=2pt]
\item If either $q_1$ or $q_2$ is unlucky:

Without loss of generality, say that $q_1$ is lucky and $q_2$ is unlucky.

\begin{itemize}[label=--, topsep=3pt, parsep=1.5pt]
\item The lucky prime $q_1$ then divides $x_j^{2^t o} - 1$.
To see why this is, recall that $x_j \in \mathbb Z_N^*$ and that $q_1 - 1$ divides $r'$ since $q_1$ is lucky, so
\begin{align*}
x_j^{2^t o} = x_j^{r'} \equiv 1 \quad (\text{mod } q_1).
\end{align*}

\item The unlucky prime $q_2$ divides $x_j^{2^t o} - 1$ iff $x_j^{2^t o} \equiv 1 \:\: (\text{mod } q_2)$.

For $x_j$ selected uniformly at random from $\mathbb Z_N^*$, and odd $q_2$, where we recall that we assumed $N$ and hence $q_2$ to be odd in the introduction, this event occurs with probability at most $1/2$.

To see why this is, note that since $q_2$ is unlucky, only an element $x_j$ with an order modulo $q_2$ that is reduced from the maximum order $q_2 - 1$ by some factor dividing $q_2 - 1$ can fulfill the condition.
The reduction factor must be at least two.
It follows that at most $1/2$ of the elements in $\mathbb Z_N^*$ can fulfill the condition.
\end{itemize}

For each iteration $j \in [1, k]$, the failure probability is hence at most $1/2$.

Since there are $k$ iterations the total failure probability is at most $2^{-k}$.

\item If both $q_1$ and $q_2$ are lucky:

In this case, both $q_1$ and $q_2$ divide $x_j^{2^t o} - 1$, since $x_j \in \mathbb Z_N^*$, and since $r' = 2^t o$ where $q_1 - 1$ and $q_2 - 1$ both divide $r'$, so 
\begin{align*}
x_j^{2^t o} = x_j^{r'} \equiv 1 \quad (\text{mod } q) \quad \text{ for } q \in \{q_1, q_2\}.
\end{align*}

The algorithm fails iff $x_j^{o}$ has the same order modulo both $q_1$ and $q_2$.

\vspace{1mm}

To see why this is, note that 
\begin{align*}
d_{i, j} = \gcd(x_j^{2^i o} - 1, N)
\end{align*}
is computed in step~\ref{alg:step:take-gcd} of the algorithm, for $i \in [0, t]$, and that the prime $q \in \{q_1, q_2\}$ divides $d_{i, j}$ iff $x_j^{2^i o} \equiv 1 \:\: (\text{mod } q)$.
It is only if this occurs for the same $i$ for both $q_1$ and $q_2$ that $q_1$ and $q_2$ will not split apart, i.e.~if $x_j^o$ has the same order modulo both $q_1$ and $q_2$.

To analyze the probability of $x_j^o$ having the same order modulo~$q_1$ and~$q_2$, we let~$2^{t_1}$ and~$2^{t_2}$ be the greatest powers of two to divide~$q_1 - 1$ and~$q_2 - 1$, respectively.
Recall furthermore that we assumed $N$, and hence~$q_1$ and~$q_2$, to be odd in the introduction.
This implies that we may assume that $t \ge t_1 \ge t_2 \ge 1$ without loss of generality.

Further\-more, we shall use that $x_j$ is selected uniformly at random from 
\begin{align*}
\mathbb Z^*_{N}
\simeq
\mathbb Z^*_{p_1^{e_1}}
\times 
\ldots
\times
\mathbb Z^*_{p_n^{e_n}},
\quad
\text{ where }
\quad
q_1, q_2 \in \{ p_1, \ldots, p_n \},
\end{align*}
which implies that $x_j \text{ mod } q_1$ and $x_j \text{ mod } q_2$ are selected independently and uniformly at random from $\mathbb Z_{q_1}^*$ and $\mathbb Z_{q_2}^*$.

Consider $x_j^{2^{t_1-1} o}$:

\begin{itemize}[label=--, topsep=3pt, parsep=1.5pt]
\item
If $t_1 = t_2$, the probability that $x_j^{2^{t_1-1} o} - 1$ is divisible by $q_1$ but not by $q_2$ is $1/4$, and vice versa for $q_2$ and $q_1$.
Hence, the probability is at most $1/2$ that $x_j^{o}$ has the same order modulo both $q_1$ and $q_2$.

\item
If $t_1 > t_2$, the probability that $x^{2^{t_1-1} o} - 1$ is divisible by $q_1$ is $1/2$, whereas the same always holds for $q_2$.
Hence, the probability is again at most $1/2$ that $x_j^{o}$ has the same order modulo $q_1$ and $q_2$.
\end{itemize}

For each iteration $j \in [1, k]$, the probability is hence again at most $1/2$.

Since there are $k$ iterations the total failure probability is $2^{-k}$.
\end{romanlist}
The lemma follows from the above argument, as there are $\binom{n}{2}$ combinations with probability at most $2^{-k}$ each.
\end{proof}

By definition $q$ is said to divide $u$ iff $u \equiv 0 \:\: (\text{mod } q)$.
Note that this implies that all $q \neq 0$ divide $u = 0$.
This situation arises in the above proof of Lemma~\ref{lemma:at-most-one-prime-unlucky}.

\begin{lemma}
\label{lemma:at-least-two-primes-unlucky}
At least two primes are unlucky with probability at most
\begin{align*}
\frac{1}{2c^{2} \log^{2} cm}.
\end{align*}
\end{lemma}
\begin{proof}
The events of various primes being unlucky are independent.
Hence, by Lemma~\ref{lemma:probability-pi-unlucky}, we have that the probability of at least two primes being unlucky is upper-bounded by 
\begin{align*}
\sum_{(i_1, i_2) \, \in \, \mathcal S} \frac{\log p_{i_1}}{m' \log m'} \cdot \frac{\log p_{i_2}}{m' \log m'}
\le
\frac{1}{2 (m' \log m')^2} \left(\, \sum_{i \, = \, 1}^n \log p_i \right)^2
\le
\frac{1}{2 c^{2} \log^{2}{cm}}
\end{align*}
where we used that $\sum_{i \, = \, 1}^n \log p_i \le \log N \le m$ and $m' = cm$, and where $\mathcal S$ is the set of all pairs $(i_1, i_2) \in [1,n]^2$ such that the product $p_{i_1} \cdot p_{i_2}$ is distinct, and so the lemma follows.
\end{proof}

\subsubsection{Runtime analysis}
\begin{claim}
\label{claim:size-r-prime}
It holds that $\log r' = \ordo(m)$.
\end{claim}
\begin{proof}
By the prime number theorem, there are $\ordo(m' / \ln m')$ primes less than $m'$.
As $r < N$ we have $\log r < m$.
Furthermore, as each prime power $q^e$ in $r' / r$ is less than $m'$, we have 
\begin{align*}
\log r' \le \log r + \ordo(m' / \ln m') \cdot \log m' = \ordo(m)
\end{align*}
as $m' = cm$ for some constant $c \ge 1$, and so the claim follows.
\end{proof}

\subsubsection{Main theorem}
\begin{theorem}
The factoring algorithm, with the possible exception of the single order finding call, completely factors $N$ in polynomial time, except with probability at most
\begin{align*}
2^{-k} \cdot \binom{n}{2} + \frac{1}{2 c^2 \log^2 cm}
\end{align*}
where $n$ is the number of distinct prime factors of $N$, $m$ is the bit length of $N$, $c \ge 1$ is a constant that may be freely selected, and $k$ is the number of iterations performed in the classical post-processing.
\end{theorem}
\begin{proof}
It is easy to see that the non-order finding part of the algorithm runs in polynomial time in $m$, as all integers are of length $\ordo(m)$, including in particular $r'$ by Claim~\ref{claim:size-r-prime}.
The theorem then follows from the analysis in section~\ref{section:success-probability}, by summing the upper bound on the probability of a failure occurring when at most one prime is unlucky in Lemma~\ref{lemma:at-most-one-prime-unlucky}, and on the probability of at least two primes being unlucky in Lemma~\ref{lemma:at-least-two-primes-unlucky}.
\end{proof}

By the above main theorem, the algorithm will be successful in completely factoring $N$,
if the constant~$c$ is selected so that~$1 / (2 c^{2} \log^{2} cm)$ is sufficiently small,
and if $2^{-k}$ for $k$ the number of iterations is sufficiently small in relation to~$\binom{n}{2}$ for~$n$ the number of distinct prime factors in~$N$.
Note that the latter requirement is easy to meet:
Pick $k = (2+\tau) \log n$ for some $\tau \ge 1$.
Then $2^{-k} \cdot \binom{n}{2} \le n^{-\tau}$ where $n \ge 2$ for composite $N$.

The time complexity of the algorithm is dominated by $k$ exponentiations of an integer modulo $N$ to an exponent of length $\ordo(m)$ bits.
This is indeed very efficient.
Note furthermore that our analysis of the success probability of the algorithm is a worst case analysis.
In practice, the actual success probability of the algorithm is higher.
Also, nothing in our arguments strictly requires~$c$ to be a constant:
We can make $c$ a function of $m$ to further increase the success probability at the expense of working with $\ordo(cm)$ bit exponents.

\section{Summary and conclusion}
When factoring an integer $N$ via order finding, as in Shor's factoring algorithm, computing the order of a single element selected uniformly at random from $\mathbb Z^*_N$ suffices to completely factor $N$, with very high probability, depending on how $c$ and $k$ are selected in relation to the number of factors $n$ and the bit length $m$ of $N$, for $N$ any integer.

\section*{Acknowledgments}
I am grateful to Johan Håstad for valuable comments and advice.
Funding and support for this work was provided by the Swedish NCSA that is a part of the Swedish Armed Forces.

I thank the participants of the Schloss Dagstuhl quantum cryptanalysis seminar, and in particular Daniel J.~Bernstein, for asking questions eventually leading me to consider more general factoring problems.

\appendix
\section{Supplementary simulations}
We have implemented the algorithm in Sage and tested it in practice.
Note that this is possible for any problem instance for which the factorization of $N = p_1^{e_1} \cdot \ldots \cdot p_n^{e_n}$ is known:

Order finding can be implemented exactly classically if the factorization of $p_i - 1$ is known for all $i \in [1, n]$.
If only the $p_i$ are known, order finding can be simulated heuristically classically:

The correct order of $g$ is then returned with very high probability.
If the correct order is not returned, some positive multiple of the correct order is returned, see section~\ref{section:simulate-order-finding}.

\subsection{Selecting problem instances}
To setup problem instances for our tests, we select 
$n$ distinct primes $\{ p_1, \, \ldots, \, p_n \}$ uniformly at random from the set of all odd $\ell$ bit primes,
exponents $\{ e_1, \, \ldots, \, e_n \}$ uniformly at random from the integers on $[1, e_{\max}]$, and 
compute the product $N = p_1^{e_1} \cdot \ldots \cdot p_n^{e_n}$.

\subsection{Simulating order finding}
\label{section:simulate-order-finding}
To simulate order finding heuristically, we use for computational efficiency that selecting $g$ uniformly at random from $\mathbb Z_N^*$ is equivalent to selecting $(g_1, \ldots, g_n)$ uniformly at random from
\begin{align*}
\mathbb Z^*_{p_1^{e_1}}
\times
\ldots
\times
\mathbb Z^*_{p_n^{e_n}}
\simeq
\mathbb Z_N^*.
\end{align*}

To approximate the order $r_i$ of each $g_i$ thus selected, we let $\phi(p_i^{e_i})$ be an initial guess for~$r_i$.
For all $f \in \mathcal P(B_s)$ for some bound $B_s$, we then let $r_i \leftarrow r_i / f$ for as long as $f$ divides $r_i$ and
\begin{align*}
g_i^{r_i/f} \equiv 1 \quad (\text{mod } p_i^{e_i}),
\end{align*}
where we recall that $\mathcal P(B_s)$ is the set of all primes $\le B_s$.

We then construct $g$ from $g_i$ and $p_i$, $e_i$ via the Chinese remainder theorem, by requiring that $g \equiv g_i \:\: (\text{mod } p_i^{e_i})$ for all $i \in [1, n]$, and take $r = \lcm(r_1, \, \ldots, \, r_n)$ as the approximate order~$r$ of~$g$.
This is a good approximation of $r$, in the sense that it is equal to $r$ with very high probability, provided that the bound $B_s$ is selected sufficiently large.

There is of course still a tiny risk that the approximation of $r$ will be incorrect, in which case it will be equal to some multiple $c_r \cdot r$, for $c_r > B_s$ a factor that divides $\phi(N)$.
This implies that the factoring algorithm will perform slightly better under simulated order finding than under exact order finding, as the order is never reduced by factors greater than~$B_s$.
The difference is however negligible for sufficiently large~$B_s$.

\subsection{Results}
We have executed tests, in accordance with the above, for all combinations of
\begin{align*}
\ell \in \{ 256, \, 512, \, 1024 \} 
\quad
n \in \{ 2, \, 5, \, 10, \, 25 \}
\quad
e_{\max} \in \{ 1, \, 2, \, 3 \}
\end{align*}
with $c = 1$, unbounded $k$, and $B_s = 10^6$ in the order finding simulator.
As expected, the algorithm recovered all factors efficiently from $r$ and $N$ in all cases considered.

For these choices of parameters, the runtime typically varies from seconds up to minutes when the Sage script is executed on a regular laptop computer.

\end{document}